\def\year{2021}\relax
\documentclass[letterpaper]{article} 
\usepackage{aaai21}  
\usepackage{times}  
\usepackage{helvet} 
\usepackage{courier}  
\usepackage[hyphens]{url}  
\usepackage{graphicx} 
\usepackage{natbib}  
\usepackage{caption} 
\usepackage[]{algorithm2e}
\usepackage[colorinlistoftodos,prependcaption,textsize=tiny]{todonotes}
\usepackage[normalem]{ulem}
\usepackage[utf8]{inputenc}
\usepackage{acro}
\usepackage{adjustbox}
\usepackage{amsfonts}
\usepackage{amsmath}
\usepackage{amssymb}
\usepackage{amsthm}
\usepackage{blindtext}
\usepackage{bm}
\usepackage{centernot}
\usepackage{changepage}
\usepackage{csquotes}
\usepackage[shortlabels]{enumitem}
\usepackage{float}
\usepackage{lipsum}
\usepackage{mathtools}
\usepackage{makecell}
\usepackage{multicol}
\usepackage{ragged2e}
\usepackage{rotating}
\usepackage{scalerel}
\usepackage{stackengine}
\usepackage{subcaption}
\usepackage{tablefootnote}
\usepackage{tabularx}
\usepackage{thm-restate}
\usepackage{thmtools}
\usepackage{wasysym}
\usepackage{wrapfig}
\usepackage{xargs}
\usepackage{xspace}
\usepackage{xparse}
\usepackage{xpatch}
\usepackage{url}            
\usepackage{nicefrac}       

\graphicspath{
  {./},
  {./images/},
  {../images/},
  {../../images},
}

\title{Sound Algorithms in Imperfect Information Games}
\author{
    Michal \v Sustr\textsuperscript{\rm 1,\rm 2}~
    Martin Schmid\textsuperscript{\rm 2}~
    Matej Morav\v{c}\'{i}k\textsuperscript{\rm 2}~
    Neil Burch\textsuperscript{\rm 2}~
    Marc Lanctot\textsuperscript{\rm 2}~
    Michael Bowling\textsuperscript{\rm 2} \\
}
\affiliations {
    \textsuperscript{\rm 1} Czech Technical University~
    \textsuperscript{\rm 2} DeepMind \\
    Corresponding authors: michal.sustr@aic.fel.cvut.cz, mschmid@google.com
}

\newtheorem{theorem}{\protect\theoremname}
\newtheorem{definition}[theorem]{\protect\definitionname}
\newtheorem{lemma}[theorem]{\protect\lemmaname}

\newtheorem{example}[theorem]{\protect\examplename}
\newtheorem{corollary}[theorem]{\protect\corollaryname}
\newtheorem{remark}[theorem]{\protect\remarkname}

\newtheorem*{lemma*}{\protect\lemmaname}
\newtheorem*{definition*}{\protect\definitionname}

\providecommand{\corollaryname}{Corollary}
\providecommand{\claimname}{Claim}
\providecommand{\definitionname}{Definition}
\providecommand{\lemmaname}{Lemma}
\providecommand{\notationname}{Notation}
\providecommand{\remarkname}{Remark}
\providecommand{\problemname}{Problem}
\providecommand{\propositionname}{Proposition}
\providecommand{\examplename}{Example}
\providecommand{\theoremname}{Theorem}
\providecommand{\conjecturename}{Conjecture}
\providecommand{\experimentname}{Experiment}

\DeclareMathAlphabet{\mathpzc}{OT1}{pzc}{m}{it}
\DeclareMathSymbol{\shortminus}{\mathbin}{AMSa}{"39}

\newcommand{\mc}{\mathcal}
\newcommand{\mb}{\mathbb}

\newcommand{\expl}{\mathpzc{expl}\!}              
\newcommand{\gv}{u^{*}}                           
\newcommand{\epsv}{u^\epsilon}                    

\newcommand{\search}{\Omega}                            
\newcommand{\searchparam}{\boldsymbol{\theta}}
\newcommand{\match}{z}

\newcommand{\pl}{n}                               
\newcommand{\opp}{{\textnormal{-} n}}             
\newcommand{\br}{{\mathpzc{br}}}                  
\newcommand{\brv}{{\mathpzc{brv}}}                

\newcommand{\fig}[1]{Figure~\ref{fig:#1} }        

\newcommand{\defword}[1]{\textbf{\boldmath{#1}}}

\newcommand{\ffoogg}{factored-observations stochastic game\xspace}   
\newcommand{\ffooggs}{factored-observations stochastic games\xspace} 

\newcommand{\pub}{{\textnormal{pub}}}             
\newcommand{\priv}[1]{{\textnormal{priv(#1)}}}    

\newcommand{\worldTree}{\ensuremath{\mc H}}

\newcommand{\priHist}{\ensuremath{s_n}}
\newcommand{\priTree}{\ensuremath{\mc S_n}}

\newcommand{\fog}{\ensuremath{\mc G}}

\newcommand{\matches}{\ensuremath{m}}
\newcommand{\keps}{\ensuremath{(k,\epsilon)}}

\DeclareMathOperator{\EX}{\mathbb{E}}
\DeclareMathOperator*{\argmax}{arg\,max}

\newlist{todolist}{itemize}{2}
\setlist[todolist]{label=$\square$}
\usepackage{pifont}
\newcommand*\cleartoleftpage{%
  \clearpage
  \ifodd\value{page}\hbox{}\newpage\fi
}

\definecolor{insignwin}{rgb}{0.5,0.5,0.5}
\definecolor{insignlose}{rgb}{0.5,0.5,0.5}
\definecolor{win}{rgb}{0,0,0}
\definecolor{lose}{rgb}{0,0,0}

\usepackage{tikz} 
\usetikzlibrary{shapes} 
\usetikzlibrary{calc} 

\newlength{\nodesize}
\colorlet{chance_color}{black}
\colorlet{pl0_color}{chance_color}
\colorlet{chance_text}{white}
\colorlet{pl1_color}{magenta!50}
\colorlet{pl2_color}{green!50!lime!60}
\tikzset{
basenode/.style = {draw,
inner sep = 0.1em,
minimum size = \nodesize
},
playernode/.style={basenode,
shape = regular polygon,
regular polygon sides = 3
},
pl1/.style={playernode, fill=pl1_color},
pl2/.style={playernode, fill=pl2_color, shape border rotate=180},
chance/.style = {basenode,
fill=pl0_color, text=chance_text,
circle,
minimum size=0.7*\nodesize,
},
terminal/.style = {basenode,
draw=none,
outer sep=0,
minimum size = 0.6\nodesize
}
}

\begin{document}

\maketitle

\begin{abstract}

Search has played a fundamental role in computer game research since the very
beginning. And while online search has been commonly used in perfect information
games such as Chess and Go, online search methods for imperfect information
games have only been introduced relatively recently. This paper addresses the
question of what is a sound online algorithm in an imperfect information setting
of two-player zero-sum games. We argue that the~fixed-strategy~definitions of
exploitability and $\epsilon$-Nash equilibria are ill-suited to measure an
online algorithm's worst-case performance. We thus formalize
$\epsilon$-soundness, a concept that connects the worst-case performance of an
online algorithm to the performance of an $\epsilon$-Nash equilibrium. As
$\epsilon$-soundness can be difficult to compute in general, we introduce a
consistency framework --- a hierarchy that connects an online algorithm's
behavior to a Nash equilibrium. These multiple levels of consistency describe in
what sense an online algorithm plays ``just like a fixed Nash equilibrium''.
These notions further illustrate the difference between perfect and imperfect
information settings, as the same consistency guarantees have different
worst-case online performance in perfect and imperfect information games. The
definitions of soundness and the consistency hierarchy finally provide
appropriate tools to analyze online algorithms in repeated imperfect information
games. We thus inspect some of the previous online algorithms in a new light,
bringing new insights into their worst-case performance guarantees.

\end{abstract}

\section{Introduction}

From the very dawn of computer game research, search  was a fundamental
component of many algorithms. Turing's chess algorithm from $1950$ was able to
think two moves ahead~\cite{copeland2004essential}, and Shannon's work on chess
from $1950$ includes an extensive section on how an evaluation function can be
used within search~\cite{shannon1950xxii}. Samuel's checkers algorithm from
$1959$ already combines search and learning of a value function, approximated
through a self-play method and bootstrapping~\cite{samuel1959some}. The
combination of search and learning has been a crucial component in the
remarkable milestones where computers outperformed their human counterparts in
challenging games: DeepBlue in Chess~\cite{campbell2002deep}, AlphaGo in
Go~\cite{alphago}, DeepStack and Libratus in Poker~\cite{DeepStack, Libratus}.

Online methods for approximating Nash equilibria in sequential imperfect
information games appeared only in the last few years~\cite{OOS, brown2017safe,
DeepStack, Libratus, Pluribus, rebel}. We thus investigate what it takes for an
online algorithm to be sound in imperfect information settings. While it has
been known that search with imperfect information is more challenging than with
perfect information~\cite{frank1998search, OOS}, the problem is more complex
than previously thought. Online algorithms ``live'' in a fundamentally different
setting, and they need to be evaluated appropriately.

Previously, a common approach to evaluate online algorithms was to compute a
corresponding offline strategy by ``querying'' the online algorithm at each
state (``tabularization'' of the strategy)~\cite{OOS,MCCR}. One would then
report the exploitability of the resulting offline strategy. We show that this
is not generally possible and that naive tabularization can also lead to
incorrect conclusions about the online algorithm's worst-case performance. As a
consequence we show that some algorithms previously considered to be sound are
not.

We first give a simple example of how an online algorithm can lose to an
adversary in a repeated game setting. Previously, such an algorithm would be
considered optimal based on a naive tabularization. We build on top of this
example to introduce a framework for properly evaluating an online algorithm's
performance. Within this framework, we introduce the definition of a sound and
$\epsilon$-sound algorithm. Like the exploitability of a strategy in the offline
setting, the soundness of an algorithm is a measure of its performance against a
worst-case adversary. Importantly, this notion collapses to the previous notion
of exploitability when the algorithm follows a fixed strategy profile.

We then introduce a consistency framework, a hierarchy that formally states in
what sense an online algorithm plays ``consistently'' with an
$\epsilon$-equilibrium. The hierarchy allows stating multiple bounds on the
algorithm's soundness, based on the $\epsilon$-equilibrium and consistency type.
The stronger the consistency is in our hierarchy, the stronger are the bounds.
This further illustrates the discrepancy of search in perfect and imperfect
information settings, as these bounds sometimes differ for perfect and imperfect
information games.

The definitions of soundness and the consistency hierarchy finally provide
appropriate tools to analyze online algorithms in imperfect information games.
We thus inspect some of the previous online algorithms in a new light, bringing
new insights into their worst-case performance guarantees. Namely, we focus on
the Online Outcome Sampling (OOS)~\cite{OOS} algorithm. Consider the following
statement from the OOS publication: ``We show that OOS is consistent, i.e., it
is guaranteed to converge to an equilibrium strategy as search time increases.
To the best of our knowledge, this is not the case for any existing online game
playing algorithm\ldots' The problem is that OOS provides only the weakest of
the introduced consistencies --- local consistency. As the local consistency
gives no guarantee for imperfect information games (in contrast to perfect
information games), OOS (and potentially other locally consistent algorithms)
can be highly exploited by an adversary. The experimental section then confirms
this issue for OOS in two small imperfect information games.

\section{Background}

We present our results using the recent formalism of \ffooggs~\cite{FOG}. The
entirety of the paper trivially applies to the extensive form
formalism~\cite{osborne1994course} as well\footnote{Under the assumption the
games are perfect-recall and 1-timeable~\cite{FOG}.} (as we are only relying on
the notion of states and rewards). We believe this choice of formalism makes it
easier to incorporate our definitions in the future online algorithms, as sound
search in imperfect information critically relies on the notion of common/public
information~\cite{CFR-D,seitz2019value}. Indeed, all the recently introduced
online algorithms in imperfect information games rely on these
notions~\cite{DeepStack, Libratus, MCCR}.

\begin{definition}A \ffoogg is a tuple \[ \fog=\langle \mc N, \mc W, w^o, \mc A,
\mc T, \mc R, \mc O \rangle, \] where: \begin{itemize} \item $\mc N = \{1,\,2\}$
is a \defword{player set}. We use symbol $\pl$ for a player and $\opp$ for its
opponent. \item $\mc W$ is a set of \defword{world states} and $w^0 \in \mc W$
is a designated initial world state. \item $\mc A = \mc A_1 \times \mc A_2$ is a
space of \defword{joint actions}. The subsets $\mc A_\pl(w) \subset \mc A_\pl$
and $\mc A(w) = \mc A_1(w) \times \mc A_2(w) \subset \mc A$ specify the (joint)
actions legal at $w \in \mc W$. For $a\in \mc A$, we write $a=(a_1,a_2)$. $\mc
A_\pl(w)$ for $n \in \mc N$ are either all non-empty or all empty. A world state
with no legal actions is \defword{terminal}. \item After taking a (legal) joint
action $a$ at $w$, the \defword{transition function} $\mc T$ determines the next
world state $w'$, drawn from the probability distribution $\mc T(w,a) \in \Delta
(\mc W)$. \item $\mc R = (\mc R_1, \mc R_2)$, and $\mc R_\pl(w,a)$ is the
\defword{reward} player $\pl$ receives when a joint action $a$ is taken at $w$.
\item $\mc O = (\mc O_\priv{1},\mc O_\priv{2}, \mc O_\pub )$ is the
\defword{observation function}, where $\mc O_{(\cdot)} : \mc W \times \mc A
\times \mc W \to \mb O_{(\cdot)} $ specifies the \defword{private observation}
that player $\pl$ receives, resp. the \defword{public observation} that every
player receives, upon transitioning from world state $w$ to $w'$ via some joint
action $a$. \end{itemize} \end{definition}

A legal \defword{world history} (or trajectory) is a finite sequence $h = (w^0,
a^0, w^1, a^1, \,\dots\, , w^t)$, where $w^k \in \mc W$, $a^k \in \mc A(w^k)$,
and $w^{k+1} \in \mc W$ is in the support of $\mc T(w^k, a^k)$. We denote the
set of all legal histories by $\worldTree$, and the set of all sub-sequences of
$h$ that are legal histories as $\worldTree(h) \subseteq \worldTree$.

Since the last world state in each $h\in \worldTree$ is uniquely defined, the
notation for $\mc W$ can be overloaded to work with $\worldTree$ (e.g., $\mc
A(h) := \mc A(\textnormal{the last $w$ in $h$})$, $h$ being terminal, ...). We
use $\mc Z$ to denote the set of all terminal histories, i.e. histories where
the last world state is terminal.

The \defword{cumulative reward} of $\pl$ at $h$ is $\sum_{k=0}^{t-1} r^k_\pl :=
\sum_{k=0}^{t-1} \mc R_\pl(w^k, a^k)$. When $h$ is a terminal history,
cumulative rewards can also be called \defword{utilities}, and denoted as
$u_\pl(z)$. We assume games are \defword{zero-sum}, so $u_\pl(z) = -u_\opp(z) \;
\forall z \in \mc Z$. The maximum difference of utilities is $\Delta = \left|
\max_{z \in \mc Z} u_1(z) - \min_{z \in \mc Z} u_1(z) \right|$

Player $\pl$'s \defword{information state} or \defword{private history} at $h =
(w^0, a^0, w^1, a^1,\,\dots\,, w^t)$ is the action-observation sequence
$ \priHist(h) := (O^0_\pl, a^0_\pl, O^1_\pl, a^1_\pl,\,\dots\,,
O^t_\pl)$, where $O^k_\pl = \mc O_\pl(w^{k-1}, a^{k-1}, w^k)$ and $O^0_\pl$ is
some initial observation. The space $\priTree$ of all such sequences can be
viewed as the \defword{private tree} of $\pl$.

A \defword{strategy profile} is a pair $\sigma = (\sigma_1,\sigma_2)$, where
each (behavioral) \defword{strategy} $\sigma_\pl : \priHist \in \priTree \mapsto
\sigma_\pl(\priHist) \in \Delta(\mc A_\pl(\priHist))$ specifies the probability
distribution from which player $\pl$ draws their next action (conditional on
having information $\priHist$). We denote the set of all strategies of player
$\pl$ as $\Sigma_\pl$ and the set of all strategy profiles as $\Sigma$.

The~\defword{reach probability} of a history $h\in \worldTree$ under $\sigma$ is
defined as $
\pi^{\sigma}(h)=\pi^{\sigma}_{1}(h)\,\pi^{\sigma}_{2}(h)\,\pi^\sigma_c(h), $
where each $\pi^{\sigma}_{n}(h)$ is a~product of probabilities of the~actions
taken by player $n$ between the root and $h$, and $\pi^{\sigma}_c(h)$ is the
product of stochastic transitions. The~\defword{expected utility} for player $n$
of a strategy profile $\sigma$ is $u_n(\sigma) = \sum_{z\in \mc Z}
\pi^\sigma(z)u_n(z)$.

We define a \defword{best response} of player $n$ to the other player's
strategies $\sigma_{\opp}$ as a strategy $\br(\sigma_{\opp}) \in
\argmax_{\sigma_n' \in \Sigma_n} u_n(\sigma_n', \sigma_{\opp} )$ and
\defword{best response value} $\brv(\sigma_{\opp}) = \max_{\sigma_n' \in
\Sigma_n} u_n(\sigma_n',\,\sigma_{\opp} )$. The~profile $\sigma$ is an
\defword{$\epsilon$-Nash equilibrium} if $\left( \forall \, n \in \mc N \right)
\ : \ u_n(\sigma) \geq \max_{\sigma_n' \in \Sigma_n} u_n(\sigma_n', \sigma_\opp)
- \epsilon$, and we denote the set of all $\epsilon$-equilibrium strategies of
player $n$ as $\mc{NE}^{\epsilon}_n$. The \defword{strategy exploitability} is
$\expl_n(\sigma_n) := \big[ u_n(\sigma^*) - \min_{\sigma'_\opp \in \Sigma_\opp}
u_n(\sigma_n,\sigma'_\opp) \big]$ where $\sigma^*$ is an equilibrium strategy.
The \defword{game value} $\gv = u_1(\sigma^*)$ is the utility player 1 can
achieve under a Nash equilibrium strategy profile.

\section{Online Algorithm}\label{sec:online-algorithm}

The environment we are concerned with is that of a repeated game, consisting of
a sequence of individual matches. As a match progresses, the algorithm produces
a strategy for a visited information state on-line: that is, once it actually
observes the state. This common framework of repeated games is particularly
suitable for analysis of online algorithms, as the online algorithm can be
conditioned on the past experience (e.g. by trying to adapt to the opponent or
by re-using parts of the previous computation). We are then interested in the
accumulated reward of the agent during the span of the repeated game. Of
particular interest will be the expected reward against a worst-case adversary.

\subsection{Coordinated Matching Pennies}

We now introduce a small imperfect information game that will be used throughout
the article -- ``Coordinated Matching Pennies''~(CMP). It is a variation on the
well-known Matching Pennies game~\cite{osborne1994course}, where players choose
either Heads or Tails and receive a utility of $\pm 1$ if their actions
(mis)match. For CMP, we additionally introduce a publicly observable chance
event just after the first player acts. See \fig{matching-pennies}for details.

\begin{figure}[t]
  \centering
  \includegraphics[width=0.6\linewidth]{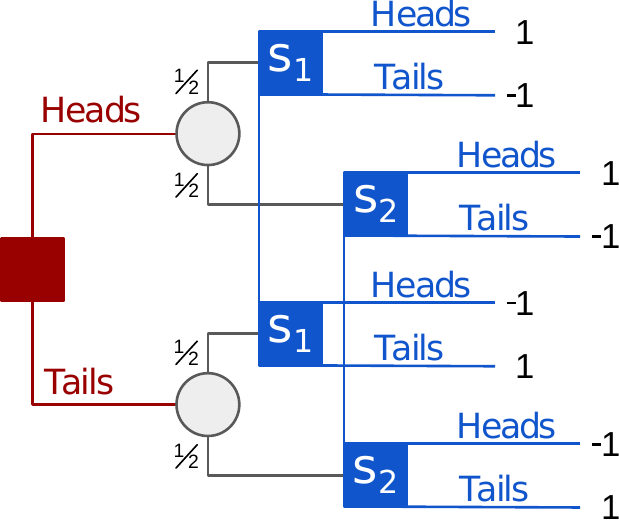}

  \caption{Coordinated Matching Pennies. After the first player acts, chance
  randomly chooses whether the second player will be playing in the information
  state $s_1$ or $s_2$. The first player receives utility of $1$ if players'
  actions match and $-1$ if they mismatch. }

  \label{fig:matching-pennies}

\end{figure}

Let $p$ and $q$ denote the probability of playing Heads in information states
$s_1$ and $s_2$ respectively. An equilibrium strategy for the second player
(Blue) is then any strategy where the average of $p$ and $q$ is $\frac{1}{2}$.
He thus has to coordinate the actions between his two information states, while
the first player has a unique uniform equilibrium strategy. Similar equilibrium
coordination happens also in Kuhn Poker~\cite{kuhn1950simplified}.

\subsection{Naive Tabularization} We now show that if one naively tries to
convert an online algorithm into a fixed strategy, the resulting exploitability
is not always representative of the worst-case performance of the online
algorithm. Consider the following algorithm \texttt{PlayCache} for the repeated
game of CMP. \texttt{PlayCache} keeps an internal state, a cache -- a mapping of
information state to probability distribution over the actions, and it gradually
fills the cache during the game.

Concretely, \texttt{PlayCache} plays for the second player as follows:

\begin{itemize}
\item Initialize algorithm's state $\searchparam_0$ to an empty cache.

\item Given an information state $s$ visited during a game, there are three
possible cases: i) The cache is empty: play Heads and store $\{s,
\text{Heads}\}$ into the cache. ii) The cache is non-empty and contains $s$:
play the cached strategy for $s$. iii) The cache is non-empty and does not
contain $s$: play Tails and store $\{s, \text{Tails}\}$.

\end{itemize}

Consider what happens if one tries to naively tabularize the \texttt{PlayCache}
by querying the algorithm for all the information states. If we query the
algorithm for states $s_1,\ s_2$, we get the resulting offline strategy $s_1 :
\text{Heads},\ s_2 : \text{Tails}$. Querying the algorithm for states in reverse
order, i.e. $s_2,\ s_1$ results in $s_1 : \text{Tails},\ s_2 : \text{Heads}$.
And while both of these offline strategies have zero exploitability, one can
still exploit the algorithm during the repeated game. This follows from the fact
that the very first time the \texttt{PlayCache} gets to act, it always plays
$\text{Heads}$. The first player can thus simply play $\text{Heads}$ during the
first match and is guaranteed to win the match. As we will show later,
\texttt{PlayCache} falls within a class of algorithms that can be exploited, but
where the average reward is guaranteed to converge to the game value as we
repeatedly keep playing the game.

Where did this discrepancy between the exploitability of the tabularized
strategy and the exploitability of the online algorithm come from? It is simply
because the tabularized strategy does not properly describe the game dynamics of
\texttt{PlayCache}. In fact, there is no fixed strategy that does so! We will
now formalize an appropriate framework to describe the rewards and dynamics of
online algorithms, which will allow us to define notions for the expected reward
and the worst-case performance in the online setting.

\subsection{Online Settings}

The \defword{repeated game} $p$ consists of a finite sequence of $k$ individual
matches $\matches = (\match_1,\,\match_2,\, \hdots,\, \match_k)$, where each
\defword{match} $\match_i \in \mc Z$ is a sequence of world states and actions
$\match_i = (w_i^0,\,a_i^0,\,w_i^1,\,a_i^1\,\hdots,\,a_i^{l_i-1},  w_i^{l_i})$,
ending in a terminal world state $w_i^{l_i}$. For each visited world state in
the match, there is a corresponding information state, i.e. a player's private
perspective of the game (for perfect information games, the notion of
information state and world state collapse as the player gets to observe the
world perfectly). An online algorithm $\search$ then simply maps an information
state observed during a match to a strategy, while possibly using its internal
algorithm state (Def.~\ref{defn:online_algorithm}).

Given two players that use algorithms $\search_1, \search_2$, we use
$P^k_{\search_1, \search_2}$ to denote the distribution over all the possible
repeated games $\matches$ of length $k$ when these two players face each other.
The average reward of $\matches$ is $\mc R_\pl(\matches) = 1/k \sum_{i=1}^k
u_\pl(z_i)$ and we denote $\EX_{\matches \sim P^k_{\search_1, \search_2}}[\mc
R_\pl(\matches)]$ to be the expected average reward when the players play $k$
matches. From now on, if player $n$ is not specified, we assume without loss of
generality it is player $1$. The proofs of the theorems can be found in the
Appendix.

\begin{definition} \label{defn:online_algorithm} Online algorithm $\search$ is a
function $\priTree \times \Theta \mapsto \Delta(\mc A_\pl(\priHist)) \times
\Theta$, that maps an information state $\priHist \in \priTree$ to the strategy
$\sigma_\pl(\priHist) \in \Delta(\mc A_\pl(\priHist))$, while possibly making
use of algorithm's state $\searchparam \in \Theta$ and updating it. We denote
the algorithm's \defword{initial state} as $\searchparam_0$. A~special case of
an online algorithm is a \defword{stateless algorithm}, where the output of the
function is independent of the algorithm's state (thus independent of the
previous matches). If the output depends on the algorithm's state, we say the
algorithm is \defword{stateful}. \end{definition}

As the game progresses, the online algorithm produces strategies for the visited
information states and updates its algorithm state. This allows it to
potentially output different strategies for the same information state visited
in different matches. We thus use $\search^\matches(s_\pl)$ to denote the
resulting strategy in the information state $s_\pl$ after the algorithm has
already played the matches $\matches={z_1,\, \hdots,\, z_k}$. Note that players
can not visit the same information state twice in a single match.

\begin{remark} \label{rem:stochastic} If we need to encode a stochastic
algorithm, we can do it formally as taking the initial state to be a realization
of a random variable. The initial state should be extended to encode how the
algorithm should act (seemingly) randomly in any possible game-play situation
beforehand. \end{remark}

\subsection{Soundness of Online Algorithm}

We are now ready to formalize the desirable properties of an online algorithm in
our settings. Exploitability, resp. $\epsilon$-equilibrium considers the
expected utility of a fixed strategy against a worst-case adversary in a single
match. We thus define a similar concept for the settings of an online algorithm
in a repeated game: $\keps$-soundness. Intuitively, an online algorithm is
$\keps$-sound if and only if it is guaranteed the same reward as if it followed
a fixed $\epsilon$-equilibrium after $k$ matches.

\begin{definition}
\label{def:epsilon_sound}
For an $\keps$-sound online algorithm $\search$, the expected average reward against any opponent is at least as good as if it followed an 
$\epsilon$-Nash equilibrium fixed strategy $\sigma$ for any number of matches $k'$:
\begin{align}
\label{eq:epsilon_sound}
    \forall k' \geq k \ \forall \search_2 \, : \, \EX_{\matches \sim P^{k'}_{\search, \search_2}}\left[\mc R(\matches)\right] \geq \EX_{\matches \sim P^{k'}_{\sigma, \search_2}}\left[\mc R(\matches)\right].
\end{align}
If algorithm $\search$ is $\keps$-sound for $\forall k \geq 1$, we say the algorithm is $\epsilon$-sound.
\end{definition}

Note that this definition guarantees that an online algorithm that simply
follows a fixed $\epsilon$-equilibrium is $\epsilon$-sound. And while the online
algorithm can certainly play as a fixed strategy, online algorithms are far from
limited to doing so, e.g. \texttt{PlayCache}
from~Section~\ref{sec:online-algorithm}. \texttt{PlayCache} is $1$-sound
($\epsilon=1$) as this algorithm is highly exploitable in the first match.
Additionally, an online algorithm may be sound ($\epsilon=0$), but there might
not be any offline equilibrium that produces the same distribution of matches.

\subsection{Response Game}

To compute the expected reward $\EX_{\matches \sim P^{k'}_{\search,
\search_2}}\left[\mc R(\matches)\right]$ as in Def. (\ref{def:epsilon_sound}),
we construct a repeated game~\cite{osborne1994course} in the FOSG formalism,
where we replace the decisions of the online algorithm with stochastic (chance)
transitions. As we allow the online algorithm to be stateful and thus produce
strategies depending on the game trajectory, the response game must also reflect
this possibility. The resulting game $\fog^k_\search$ is thus exponential in
size as it reflects all possible trajectories of $k$ matches. We call this
single-player game a \defword{$k$-step response game}.

The $k$-step response game allows us to compute the best response value of a
worst-case adversary in $k$-match game-play. We will use overloaded notation
$\brv(\fog^k_\search)$ to denote this value.

\begin{restatable}{theorem}{responsegamesound} If $\forall k'\geq k \quad
\brv(\fog^{k'}_\search) \leq k'\epsilon$, then algorithm $\search$ is
$\keps$-sound. \end{restatable}

\begin{proof} If we used a fixed $\epsilon$-equilibrium strategy $\sigma$ in
each match (repetition) of a response game $\fog_\sigma^{k'}$, then the
$\brv(\fog_\sigma^{k'}) = k'\epsilon$ because adversary can gain at most
$\epsilon$ in each match. Since $\epsilon$-sound algorithm should play at least
as well as some offline $\epsilon$-equilibrium, it must have
$\brv(\fog_\Omega^{k}) \leq {k}\epsilon \ \forall k \geq 1$. For a $\keps$-sound
algorithm we add the condition of $\forall k' \geq k$. \end{proof}

\subsection{Tabularized Strategy}

When an online algorithm produces the same strategy for an information state
regardless of the previous matches, there is no need for the $k$-response game.
Fixed strategy notion sufficiently describes the behavior of the online
algorithm and thus the exploitability of the fixed strategy matches the
soundness. To compute this fixed strategy, one simply queries the online
algorithm for all the information states in the game.

\section[Relating soundness and Nash]{Relating $\keps$-Soundness and $\epsilon$-Nash}
\label{sec:consistency}

Unfortunately, our notion of $\keps$-soundness is often infeasible to reason
about, as it requires checking that the algorithm does not make strategy errors
for $\forall k'\geq k$. In this section, we introduce the concept of consistency
that allows one to formally state that the online algorithm plays
``consistently'' with an $\epsilon$-equilibrium. Our consistency notion allows
us to directly bound the $\keps$-soundness of an online algorithm. We introduce
three hierarchical levels of consistency, with varying restrictions and
corresponding bounds. Notice that they differ mainly in the
order~of~quantifiers.


\definecolor{consA}{rgb}{0.0, 0.0, 0.0}
\definecolor{consB}{rgb}{0.0, 0.0, 0.0}
\definecolor{consC}{rgb}{0.0, 0.0, 0.0}

\subsection{Local Consistency}\label{sec:local-consistency} Local consistency
simply guarantees that every time we query the online algorithm, there is an
$\epsilon$-equilibrium that has the same local behavioral strategy $\sigma(s)$
for the queried state~$s$.

\begin{definition}
\label{def:local_consistency}
Algorithm $\search$ is locally consistent with $\epsilon$-equilibria if
\begin{align*}
  &{\color{consA}\forall k\ \forall \matches=(z_1, \, z_2, \, \hdots,\, z_k) }
  \ \
  {\color{consB}\forall h \in \worldTree(z_k)}
  \ \
  {\color{consC}\exists \sigma \in \mc{NE}^{\epsilon}_n} \\
  &\text{holds that}\ \ \search^{(z_1,\,\hdots,\, z_{k-1})}(s(h)) = \sigma(s(h)).
\end{align*}
\end{definition}

While this suggests that the algorithm plays like some equilibrium, it is not
so, and the resulting strategy can be highly exploitable. This is because one
cannot combine local behavioral strategies from different $\epsilon$-equilibria
and hope to preserve their exploitability. In another perspective, as soon as
one starts to condition the selection of the strategy on private information, it
risks computing strategies that can be exploited in a repeated game. This is a
motivation behind introducing $\keps$-soundness, as it allows us to analyze
algorithms that use such conditioning.

Consider the CMP game with two strategies $\sigma^1=\{(s_1,p=1),(s_2,q=0)\}$ and
$\sigma^2=\{(s_1,p=0.5), (s_2,q=0.5)\}$. While both strategies are equilibria,
if one plays in the states $s_1$ and $s_2$ based on the first and second
equilibrium respectively, it corresponds to an exploitable strategy
$\{(s_1,p=1),(s_2, q=0.5)\}$.

\begin{restatable}{theorem}{locallyconsistentnotsound}
\label{thm:locallyconsistentnotsound}
An algorithm that is locally consistent with $\epsilon$-equilibria might not be $\keps$-sound.
\end{restatable}

Note that this can happen even in perfect information games. 
in~Appendix~\ref{ap:examples}). Interestingly, local consistency is sufficient
if the algorithm is consistent with a subgame perfect equilibrium.

\begin{restatable}{theorem}{localsubgameperfect} In perfect information games,
an algorithm that is locally consistent with a subgame perfect equilibrium is
sound. \end{restatable}

A particularly interesting example of an algorithm that is only locally
consistent is Online Outcome Sampling~\cite{OOS} (OOS). See
Section~\ref{sec:experiments} for detailed discussion and experimental
evaluation, where we show that this algorithm can produce highly exploitable
strategies in imperfect information games.

\subsection{Global Consistency}

Local consistency guarantees consistency only for individual states. The problem
we have then seen is that the combination of these local strategies might
produce highly exploitable overall strategy. A natural extension is then to
guarantee consistency with some equilibria for all the states in combination: a
global consistency.

\begin{definition}
\label{def:global_consistency}
Algorithm $\search$ is globally consistent with $\epsilon$-equilibria~if
\begin{align*}
  &{\color{consA} \forall k\ \forall \matches=(z_1,\, z_2,\, \hdots,\, z_k) }
  \ \
  {\color{consC} \exists \sigma \in \mc{NE}^{\epsilon}_n }
  \ \
  {\color{consB} \forall h \in \worldTree(z_i) }\\
  &\text{holds that}\ \ \search^{(z_1,\,\hdots,\, z_{i-1})}(s(h)) = \sigma(s(h)) \ \text{ for } \ \forall i \in \{1,\,\dots,\,k\}.
\end{align*}

\end{definition}

\noindent However:

\begin{restatable}{theorem}{globallyconsistentnotsound}
\label{thm:globallyconsistentnotsound}
An algorithm that is globally consistent with $\epsilon$-equilibria might not be $\epsilon$-sound.
\end{restatable}

\begin{proof}
    A counter-example:
    The \texttt{PlayCache} algorithm is globally consistent, but it is not sound ($\epsilon = 0$), as we have seen that it is exploitable during the first match ($k=1$).
\end{proof}

But what if the algorithm keeps on playing the repeated game? While the global
consistency with equilibria does not guarantee soundness, it guarantees that the
expected average reward converges to the game value in the limit.

\begin{restatable}{theorem}{globalconsistencyiiglimit}
\label{thm:global_consistency_iig_limit}
For an algorithm $\Omega$ that is globally consistent with $\epsilon$-equilibria,
\begin{align}
    \forall k \ \forall \search_2 \, : \, \EX_{\matches \sim P^k_{\search, \search_2}}[\mathcal{R}(\matches)] \geq \gv - \epsilon - \frac{ \bigl| \mc S_1 \bigr| \Delta}{k}.
\end{align}
\end{restatable}

\begin{corollary}
\label{crl:global_consistency_soundness}
An algorithm $\Omega$ that is globally consistent with $\epsilon$-equilibria is $\keps$-sound as $k \rightarrow \infty$.
\end{corollary}

\subsection{Strong Global Consistency}

The problem with global consistency is that it guarantees the existence of
consistent equilibrium for any game-play \emph{after} the game-play is
generated. Strong global consistency additionally guarantees that the game-play
\emph{itself} is generated consistently with an equilibrium; and as in global
consistency, the partial strategies for this game-play also correspond to an
$\epsilon$-equilibrium. In other words, the online algorithm simply exactly
follows a predefined equilibrium.

\begin{definition}
\label{def:strong_global_consistency}
Online algorithm $\search$ is strongly globally consistent with $\epsilon$-equilibrium if
\begin{align*}
   &{\color{consC} \exists \sigma \in \mc{NE}^{\epsilon}_n }
   \ \
   {\color{consA} \forall k\ \forall \matches=(z_1,\, z_2,\, \hdots,\, z_k) }
   \ \
   {\color{consB} \forall h \in \worldTree(z_k) }\\
   &\text{holds that}\ \ \search^{(z_1,\,\hdots,\, z_{k-1})}(s(h)) = \sigma(s(h)).
\end{align*}
\end{definition}

Strong global consistency guarantees that the algorithm can be tabularized, and
the exploitability of the tabularized strategy matches $\epsilon$-soundness of
the online algorithm.

\begin{restatable}{theorem}{globalconsistencysoundness}
\label{thm:global_consistency_soundness}
Online algorithm $\search$ that is strongly globally consistent with $\epsilon$-equilibrium is $\epsilon$-sound.
\end{restatable}

Canonical examples of strongly globally consistent online algorithms are
DeepStack/Libratus. In general, an algorithm that uses a notion of safe
(continual) resolving is strongly globally consistent as it essentially
re-solves some $\epsilon$-equilibrium (albeit an unknown one) that it follows.
Another, more recent example is~ReBeL~\cite{rebel}, as it essentially imitates
CFR-D iterations in conjunction with a neural network.

\subsubsection{Proving Strong Global Consistency}

While we are not aware of an algorithm that is only globally consistent (besides
the toy \texttt{PlayCache}), reasoning about global consistency can be
beneficial for showing the strong global consistency. Doing so just based on its
definition might not be straightforward. However, proving global consistency can
be easier. If applicable, we can then use the following theorem to extend the
proof to the strong global consistency.

\begin{restatable}{theorem}{globalconsistencystateless}
\label{thm:global_consistency_stateless}
If a globally consistent algorithm is stateless then it is also strongly globally consistent.
\end{restatable}
\begin{proof}

    The definition of a stateless algorithm implies that for an information
    state $s$ the algorithm always produces the same behavioral strategy
    $\sigma(s)$ as the algorithm is deterministic (all stochasticity is encoded
    within the algorithm state $\searchparam$, see Remark~\ref{rem:stochastic}).

    This means that whatever $\epsilon$-equilibria the algorithm is globally
    consistent with is independent of the current game-play or match number.
    This allows us to swap the quantifiers from
    \begin{align*}
        & {\color{consA} \forall k\ \forall \matches=(z_1,\, z_2,\, \hdots,\, z_k)}
        \
        {\color{consC} \exists \sigma \in \mc{NE}^{\epsilon}_n }
        \
        {\color{consB} \forall i \in \{1,\,\dots,\,k\} \ \forall h \in \worldTree(z_i) }
        \ : \\
        & \search^{(z_1,\,\hdots,\, z_{i-1})}(s(h)) = \sigma(s(h))
    \end{align*}
    to
    \begin{align*}
        & {\color{consC} \exists \sigma \in \mc{NE}^{\epsilon}_n }
        \
        {\color{consA} \forall k\ \forall \matches=(z_1,\, z_2,\, \hdots,\, z_k)}
        \
        {\color{consB}\forall i \in \{1,\,\dots,\,k\} \ \forall h \in \worldTree(z_i)}
        \ : \\
        & \search^{(z_1,\,\hdots,\, z_{i-1})}(s(h)) = \sigma(s(h)).
    \end{align*}

    Using the same argument we can treat the different matches $z_i$ as an iteration over $k$, leading us to strong global consistency
    \begin{align*}
        & {\color{consC} \exists \sigma \in \mc{NE}^{\epsilon}_n}
        \
        {\color{consA} \forall k\ \forall \matches=(z_1,\, z_2,\, \hdots,\, z_k)}
        \
        {\color{consB} \forall h \in \worldTree(z_k)}
        \ : \\
        & \search^{(z_1,\,\hdots,\, z_{k-1})}(s(h)) = \sigma(s(h)).
    \end{align*}

\end{proof}

\section[Relating soundness and regret]{Relating $\keps$-Soundness and Regret}
Regret is an online learning concept that has triggered design of a family of
powerful learning algorithms. Indeed, many algorithms that approximate Nash
equilibria use regret minimization~\cite{CFR}. There is a well-known connection
between regret and the Nash equilibrium solution concept. In a zero-sum game at
time $k$, if both players' overall regret $R_k$ is less than $k \epsilon$, the
average strategy profile is a $2\epsilon$-equilibrium~\cite{CFR}. The use of $k$
in $\keps$-soundness allows us to relate it with regret, and show how it is
different from the consistency hierarchy.

\begin{corollary} Any regret minimizer with a regret bound of $R_k$ is~$(k,
\frac{R_k}{2k})$-sound. \end{corollary}

\section{Experiments}\label{sec:experiments}

A particularly interesting example of an algorithm that is only locally consistent is Online Outcome Sampling (OOS)~\cite{OOS}. We use it to demonstrate the theoretical ideas in this paper with empirical experiments.  We show that local consistency does in fact fail to result in $\epsilon$-soundness in the online setting.  The problem we demonstrate is also not specific to OOS, but in general to any adaptation of an offline algorithm to the online setting where the algorithm attempts to improve its strategy during online play.

At high level, OOS runs the offline MCCFR algorithm in the full game (while also gradually building the tree), parameterized to increase the sampling probability of the current information state.
The algorithm then plays based on the resulting strategy for that particular state.
The problem is that these individual MCCFR runs can converge to different $\epsilon$-equilibria as the MCCFR is parameterized differently in each information state.
In other words, the OOS algorithm exactly suffers from the fact that it is only locally consistent.

We use two games in our experiments: Coordinated Matching Pennies from Section~\ref{sec:online-algorithm} and Kuhn Poker~\cite{kuhn1950simplified}.  We present the Coordinated Matching Pennies results here.
See Appendix~\ref{ap:experiment} for the complete experimental details and a similar experiment for Kuhn Poker.

Within a single match of Coordinated Matching Pennies, the second player will act either in $s_1$ or $s_2$. OOS will therefore bias MCCFR samples to whichever information state that actually occurs in the match. These two situations are distinct and result in two different strategies for the whole game (including the non-visited state), similarly to the example in~Section~\ref{sec:local-consistency}. To emulate what OOS does, we parametrize MCCFR runs to bias  samples into $s_1$ and $s_2$ respectively, and initialize the regrets in $s_1, s_2$ so that the MCCFR is likely to produce diverse sets of strategies.
As MCCFR is stochastic, we average the strategies over $3\cdot10^4$ random seeds. 

In~\fig{oos-mp-small}we plot exploitability for the average strategies, and unbiased MCCFR for reference.  The two biased variants of MCCFR actually converge at a similar rate to unbiased MCCFR, confirming that OOS is locally consistent: it quickly converges to an $\epsilon$-equilibria for $s_1$ and $s_2$ individually.  However, the tabularized strategy --- the strategy OOS follows online --- is many orders of magnitude more exploitable even with hundreds of thousands of online iterations.  The problem is that adapting its strategy online at $s_1$ and $s_2$ causes it to not be globally consistent with any $\epsilon$-equilibria.

\begin{figure}[t]
  \centering
  \includegraphics[width=0.5\textwidth]{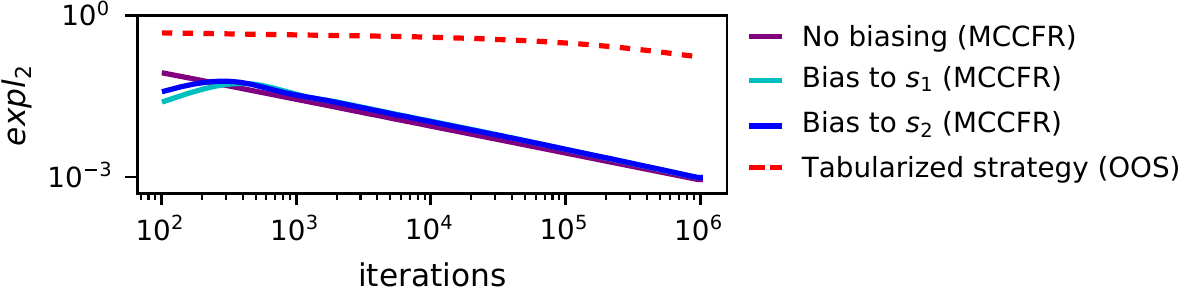}
  \caption{
  While individual MCCFR strategies have low exploitability of $\sim10^{-3}$, the~tabularized OOS strategy has high exploitability of $0.17$ even after $10^6$~iterations.
  }
  \label{fig:oos-mp-small}
\end{figure}

\section{Related literature}

There are several known pathologies that occur in imperfect information games
that are not present in the perfect information case. The pathologies that
happen in the offline setting also present a problem in the online setting.
In~\cite{frank1998search} the authors identified two problems: strategy-fusion
and non-locality.

These two problems can easily arise for algorithms designed to solve only
perfect-information games, such as minimax or reinforcement learning algorithms,
and lead to computation of exploitable strategies. The proposed local
consistency is similar in its spirit to non-locality, as composition of partial
strategies (that correspond to parts of distinct equilibria) produced by an
online algorithm may not be an overall equilibrium strategy. However local
consistency identifies sub-optimal play also across repeated games.

In~\cite{DeepStack,Libratus}, they use some form of continual re-solving, which
is strongly globally consistent. This guarantees soundness of the algorithms.
Continual resolving uses value functions defined over public belief
spaces~\cite{rebel} to compute consistent strategies. Indeed, the minimal amount
of information needed to properly define value functions are ranges (beliefs)
over common knowledge public states~\cite{seitz2019value}.

The notion of sufficient plan-time statistics studied in~\cite{Oliehoek13IJCAI}
is very closely related to the public beliefs. The paper suggests the structure
of the value function for games where the hidden information becomes public
after a certain number of moves.

We are not aware of algorithms in the literature that are only globally
consistent. This may lead to interesting future work: the algorithm may try to
reduce its sub-optimal play of the first matches, while possibly not using all
of the required player's ranges.

Tabularization has been used in~\cite{MCCR} to compute an offline strategy and
its exploitability. In~\cite{OOS} they consider computing this tabularization
(they refer to it as ``brute-force''  approach), but it is a very expensive
procedure. Instead they use an ``aggregate method'', which ``stitches'' strategy
from a small number of matches and defines the strategy as uniform in
non-visited information states. They do not state whether such approximation of
tabularization is indeed correct.

\section{Conclusion}

We introduced the game of Coordinated Matching Pennies (CMP). This game
illustrates the consistency issues that can arise for online algorithms in
imperfect information games. We observed that exploitability is not an
appropriate measure of an algorithm's performance in online settings. This
motivated us to introduce a formal framework for studying online algorithms and
allowed us to define $\epsilon$-soundness. Just like $\epsilon$-exploitability,
it measures the performance against the worst-case adversary. Soundness
generalizes exploitability to repeated sequential games and it collapses to it
when an online algorithm follows a fixed strategy. We then introduced a
hierarchical consistency framework that formalizes in what sense an online
algorithm can be consistent with a fixed strategy. Namely, we introduced three
levels of consistency: i)~local, ii)~global and iii)~strongly global. These
connect an online algorithm's behavior to that of a fixed strategy with
increasingly tight bounds on the average expected utility against a worst-case
adversary. We also stated various bounds on soundness based on the
exploitability of a consistent fixed strategy. Interestingly, the implications
are different in some cases for perfect and imperfect information games.

Within this framework, we saw that local consistency in imperfect information
games does not guarantee correct evaluation of worst-case performance by
computing exploitability. Based on this result, we argued that OOS, previously
considered sound, can be exploited. This illustrates that these subtle problems
with online algorithms can easily be missed and lead to wrong conclusions about
their performance. Our experimental section included experiments in CMP and Kuhn
Poker and showed a large discrepancy between OOS's actual performance and the
bound previously thought to hold.

\clearpage
\paragraph{Acknowledgments}

Computational resources were supplied by the project "e-Infrastruktura CZ"
(e-INFRA LM2018140) provided within the program Projects of Large Research,
Development and Innovations Infrastructures. This work was supported by Czech
science foundation grant no. 18-27483Y. We'd like to thank the anonymous
reviewers for their valuable feedback on a prior version of this paper.

\bibliography{main}

\clearpage
\appendix
\section{Consistency Examples} \label{ap:examples}

\begin{example} An online algorithm may be sound ($\epsilon = 0$), but there
might not be any offline equilibrium that produces the same distribution of
matches. \end{example}

Suppose we have a game where each player acts once, chooses from actions
$\{A,B,C\}$ and receives zero utility (i.e. a normal-form game with 3x3 zero
payoff matrix). All strategies are equilibria. If we play $k=3$ matches and the
players play pure strategies $A$, $B$ and $C$ in each match, we get a
distribution of matches $\matches=(z_1,z_2,z_3)$ that cannot be achieved with
fixed offline equilibrium. In this case, the distribution is $((w^0, (A,A)),
(w^0, (B,B)), (w^0, (C,C)))$ with probability one, and all other terminal
histories with probability zero.

\begin{example} An algorithm that is locally consistent with equilibria can be
exploited in a perfect information game. \end{example}

\begingroup
\setlength{\intextsep}{2pt}%
\setlength{\columnsep}{8pt}%
\begin{wrapfigure}{r}{0.1\textwidth}
  \centering
  \includegraphics[width=\linewidth]{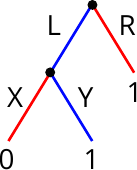}
\end{wrapfigure}

Suppose we have a single player game as in figure on the right. Both blue $L,Y$
and red $R,X$ pure strategies are equilibria. However, if the top node is
locally consistent with the blue strategy, and the bottom node with the red
strategy, the resulting strategy the algorithm actually plays is $L,X$, which is
sub-optimal.

\section{Tabularization} \label{ap:tabularization}

We can consider two ways how the online setting can be realized, with respect to how players' state changes between $k$ matches in the repeated game:
i) no-memory, where the players take turns in a match, and their memory is reset when each match is over
(players are allowed to retain memory within the individual matches), or
ii) with-memory, where the players are allowed to retain memory between the matches.

As exploitability of tabularized strategy is guaranteed to reflect $\epsilon$-soundness only for strongly globally consistent algorithms, we assume their use only. The with-memory case then collapses to the no-memory case: strongly globally consistent algorithm simply plays as some predefined (offline) equilibrium. The following text then simply defines how to compute the offline equilibrium by querying the algorithm in all states.

\newcommand{\partstrat}[2]{\sigma^{\search,\!\searchparam_{#1}}_\pl(#2)}
\begin{definition}[Partial strategy]
For a terminal history 
\[
    \match = ( w^0\!\!,\, a^0,\,w^1\!\!,\,a^1,\, \hdots,\, w^l\!\!,\,a^l,\, w^{l+1} )
\]
player $n$ has a corresponding sequence of information states\footnote{We omit the index $n$ for information state $s_n$ for clarity.}
\[
    s = ( s^0\!\!,\, s^1\!\!,\, \hdots,\, s^{l},\, s^{l+1}).
\]

We say a partial strategy $\partstrat{0}{\match}$ for player $\pl$ who uses search $\search$ and starts with state $\searchparam_0$, is an expected behavioral strategy defined only for the visited information states:
\[
    \partstrat{0}{\match} = \{ (s^t,\, \mu_t) 
        \mid 
            (\mu_t, \searchparam_{t+1}) = \EX_{\searchparam_t} \left[ \search(s^t, \searchparam_t) | s^t \right]
            \: \forall\, t \in \{0, \dots, l\}
        \}.
\]
\end{definition}

Note that when we compute the strategy $\mu_t$ from $\EX_{\searchparam_t} \left[ \search(s^t, \searchparam_t) | s^t \right]$, we must compute it as a weighted average to respect the structure of the private tree.
The weights are reach probabilities of the information state $s^t$: cumulative product of the player's strategy over the sequence of information states $s^0, \hdots s^{t-1}$, leading to target information state $s^t$. 
See~\cite[Eq.~4]{CFR} for more details.

\begin{definition}
A composition of partial strategies for terminals $\mc Z$ is a tabularized strategy
\[
    \partstrat{0}{\mc Z} = \bigcup_{\match_i \in \mc Z} \partstrat{0}{\match_i}.
\]
\end{definition}

\section{Proofs of Theorems} \label{ap:proofs}

\localsubgameperfect* \begin{proof} In perfect information games the notions of
an information state and a history blend together, as there is a one-to-one
correspondence between them. Expected utility of a history is the same for all
subgame perfect equilibria. It corresponds to the best achievable value against
worst-case adversary, given that the history occurred. This property implies
that the worst-case expected utility of a history remains optimal, if the player
plays only actions that are in the support of any subgame perfect equilibrium in
the consequent states. A formal proof can be constructed by induction on the
maximal distance from a terminal history.

Notice that this exactly happens if the player plays according to an algorithm
locally consistent with subgame perfect equilibria. The expected worst-case
utility for any history will be optimal, and the worst-case expected utility of
a match will correspond to the worst-case expected utility of the history at the
beginning of the game. Therefore the worst-case expected utility of each match
is also optimal and the algorithm is sound. \end{proof}

\newcommand{\filled}{{\bullet}}  %
\newcommand{\emptyy}{{\circ}}    %
\newcommand{\notnash}{{\times}}  %

We will now prepare the ground to prove
Thm.~\ref{thm:global_consistency_iig_limit}.

When an algorithm that is globally consistent with an $\epsilon$-equilibrium is
queried in some information states in a match in the repeated game, it will
always keep playing the same behavioral strategy in these situations in
subsequent matches. We call this as ``filling in'' strategy. Once the algorithm
fills the strategy in all player's information states, we are guaranteed to get
match reward of $\epsv = u^* - \epsilon$ on average against a worst-case
adversary.

Informally speaking, the bound in Thm.~\ref{thm:global_consistency_iig_limit}
can be easily seen to be true for a game like Coordinated Matching Pennies, or
some generalization which will have a larger number of information states that
need to be coordinated (think of ``Coordinated Rock-Paper-Scissors''). At every
match, we can incur a loss of at most $\Delta$ when reaching an unfilled
history. This is a rather pessimistic lower bound on the value, but it lets us
ignore the algorithm state: we are either playing at filled information states,
or achieving the worst possible value. The problematic part is making sure the
bound holds also when we (repeatedly) visit previously filled information
states. For each of the possible future subgames, there are two cases. In both
cases, the number of $\Delta$-sized losses in utility plus the number of
unfilled information states does not increase, so we can use induction on the
length of the game to prove the claim. Along branches where the opponent had an
opportunity to exploit the algorithm by playing into an unfilled information
state, the algorithm loses at most $\Delta$ utility compared to the equilibrium,
but must fill in at least one information state to do so. Along branches where
the agent played through filled information states, the algorithm is playing
identically to the equilibrium strategy and thus achieves the same value.

To prove Thm.~\ref{thm:global_consistency_iig_limit} we will need to establish a
Lemma~\ref{lem:bound-eq-diff}, a bound of difference of utilities a player can
gain if he plays according to a partially filled $\epsilon$-equilibrium strategy
compared to $\epsv$ within an arbitrary match. The idea of the proof for
Thm.~\ref{thm:global_consistency_iig_limit} is then to bound this difference for
any number of non-visited information states and any number of remaining matches
within the response game using induction.

An online algorithm can fill in the strategy only into information states found
on the trajectory to a terminal history, as it will be queried only in these
situations. So after playing through a match $\match = ( w^0\!\!,\,
a^0,\,w^1\!\!,\,a^1,\, \hdots,\, w^l\!\!,\,a^l,\, w^{l+1} )$, the algorithm's
response at $s_\pl(w^i)$ will be fixed as $\sigma(s_\pl(w^i))$ for all visited
worlds $w^i$ on the trajectory $\match$.

To talk about possible filled strategies within a single match, we will
partition $\mc Z$ into two non-empty sets of terminal histories $\mc Z_\filled$
(pronounced ``filled'') and $\mc Z_\emptyy$ (pronounced ``empty'' or
``unfilled''). The partition has a special property of ``being possible to
realize in online setting'': all information states on the trajectory to
terminals $\mc Z_\filled$ are filled, and all terminals that can be reached just
through these filled information states are also in $\mc Z_\filled$ (we are not
taking into consideration the opponent's information states, i.e. we operate
only on the online player's private tree).

\begin{lemma}
\label{lem:bound-eq-diff}

For a probability of reaching a filled terminal $P(\filled) = \sum_{z_\filled
\in \mc Z_\filled} \pi^{\sigma}(z_\filled)$, an expected received utility for
filled terminals $u(\filled) = \frac{\sum_{z_\filled \in \mc Z_\filled}
\pi^{\sigma}(z_\filled) u_1(z_\filled)}{\sum_{z_\filled \in \mc Z_\filled}
\pi^{\sigma}(z_\filled)}$ and an  utility of playing outside of filled histories
$u(\notnash)$, it holds that

\begin{align}
\label{eq:lemma-ineq-bound}
P(\filled)(u(\filled) - \epsv) \geq -(1-P(\filled)) (u(\notnash) - \epsv + \Delta),
\end{align}
assuming $0 < P(\filled) < 1$.
\end{lemma}

\begin{proof}

For any strategy profile $\sigma=(\sigma_1^\epsilon, \sigma_2)$ with an
$\epsilon$-equilibrium strategy $\sigma_1^\epsilon$ and arbitrary opponent
strategy $\sigma_2$ it holds that \begin{align} \label{eq:partition-eq}
\sum_{z_\filled \in \mc Z_\filled} \pi^{\sigma}(z_\filled) u_1(z_\filled) +
\sum_{z_\emptyy \in \mc Z_\emptyy} \pi^{\sigma}(z_\emptyy) u_1(z_\emptyy) \geq
\epsv. \end{align} The terms can be simplified and rewritten as factorization of
product of probabilities and (weighted) utilities as \[ P(\filled) u(\filled) =
\underbrace{\sum_{z_\filled \in \mc Z_\filled}
\pi^{\sigma}(z_\filled)}_{P(\filled)} \cdot \underbrace{\frac{\sum_{z_\filled
\in \mc Z_\filled} \pi^{\sigma}(z_\filled)
u_1(z_\filled)}{\sum\limits_{z_\filled \in \mc Z_\filled}
\pi^{\sigma}(z_\filled)}}_{u(\filled)}, \] and similarly for the ``$\emptyy$''
partition. It also holds that $P(\filled) + P(\emptyy) = 1$, as the probability
of reaching a terminal history within a match is equal to one.

We can restate \eqref{eq:partition-eq} as
\begin{align}
\label{eq:partition-eq-diff}
P(\filled) ( u(\filled) - \epsv) + (1-P(\filled) ( u(\emptyy) - \epsv) \geq 0.
\end{align}

Suppose that for the partition ``$\emptyy$'' we didn't use an equilibrium
strategy for player 1, but arbitrary strategy profile $\sigma'$ satisfying
$P^{\sigma}(\filled) + P^{\sigma'}(\emptyy) = 1$. We will denote its utility
\begin{align}
\label{eq:ux-def}
u(\notnash)=\frac{\sum_{z_\emptyy \in \mc Z_\emptyy} \pi^{\sigma'}(z_\emptyy) u_1(z_\emptyy)}{\sum_{z_\emptyy \in \mc Z_\emptyy} \pi^{\sigma'}(z_\emptyy)}.
\end{align}

The value of any two strategies cannot differ by more than the maximum
difference of utilities in the game:
\begin{align}
    \label{eq:util-eq-ne}
    u(\emptyy) \leq u(\notnash) + \Delta.
\end{align}

Putting \eqref{eq:util-eq-ne} back to \eqref{eq:partition-eq-diff}, we get the
lemma that lower bounds the difference of filled partition and $\epsv$ for an
arbitrary match:
\begin{align}
P(\filled)(u(\filled) - \epsv) \geq -(1-P(\filled)) (u(\notnash) - \epsv + \Delta).
\end{align}

\end{proof}

\globalconsistencyiiglimit*
\begin{proof}
Let us rewrite the theorem slightly:
\[
    \forall k \ \forall \search_2 \, : \, k \EX_{\matches \sim P^k_{\search, \search_2}}[\mathcal{R}(\matches)] - k\epsv \geq - \bigl| \mc S_1 \bigr| \Delta.
\]
Since $\mc R(\matches)$ is average reward, multiplying by $k$ we get cumulative utilities in the game-play $\matches=(z_1, z_2,\,\dots,\,z_k)$:
\begin{align}
  \label{eq:thm-gc-diff}
    \forall k \ \forall \search_2 \, : \, \EX_{\matches \sim P^k_{\search, \search_2}}\left[ \sum_{i=1}^k u_1(z_i) \right] - k\epsv \geq - \bigl| \mc S_1 \bigr|  \Delta.
\end{align}

So on the left side of the inequality we have a difference of cumulative
(expected) utilities and of cumulative $\epsv$. We use cumulative values because
we are now in the setting of a $k$-repeated game.

Let $v$ be the number of non-visited information states of player 1 (resp. the
number of unfilled information states) in a match, i.e. $0 \leq v \leq |\mc
S_1|$, and let $l$ be the number of next matches (including the current one),
i.e. $1 \leq l \leq k$. We will use $a_{v,l}$ to denote the difference between
expected cumulative rewards and cumulative $\epsv$ from the current match
(inclusively) until the end of the game, if we are playing against worst-case
adversary. The left side~of~\eqref{eq:thm-gc-diff} corresponds to a value equal
or greater than $a_{|\mc S_1|, k}$, so we need to prove that $a_{|\mc S_1|, k}
\geq -\bigl| \mc S_1 \bigr|  \Delta$. It is sufficient to consider only the
worst-case adversary, as the bound on $a_{v,l}$ will hold for any other opponent
as well.

We will prove the theorem by induction on $a_{v,l}$ using $v$ and $l$
simultaneously. Let us characterize the \emph{base case}. If we have visited all
information states ($v=0$), we filled $\epsilon$-equilibrium strategy
everywhere. So at each visit of such a match we receive a reward of $\epsv$, and
the difference between expected cumulative rewards and cumulative $\epsv$ is
zero:
\begin{align}
\label{eq:base-case}
    a_{0,l} = 0 \quad \forall l.
\end{align}

The \emph{induction hypothesis} is 
\begin{align}
\label{eq:ind-hypo}
    a_{x,y} \geq -x\Delta \quad \forall x \leq v \ \forall y < l.
\end{align}

There are two possibilities of what can happen in a match. We either ``hit'' the
filled information states, receive some (expected) reward $u(v,l)$ and possibly
continue into next match where we receive $a_{v,l-1}$ (if the current match is
not the last one, i.e. $l > 1$). Or we ``miss'' the filled information states,
meaning we visit arbitrary number of new information states previously not
visited. This will also change $v$ to be smaller for all subsequent matches.

We state this with an abuse of notation as
\begin{equation}
\begin{aligned}
\label{eq:avl}
    a_{v,l} &= P(v,l)(u(v,l)-\epsv + a_{v,l-1}) \\ 
    &+ P(v-1,l)(u(v-1,l)-\epsv + a_{v-1,l-1}) \\
    &+ P(v-2,l)(u(v-2,l)-\epsv + a_{v-2,l-1}) \\
    &+\ \dots \\
    &+ P(v-v,l)(u(v-v,l)-\epsv + a_{v-v,l-1}),
\end{aligned}
\end{equation}
where the terms $P(v-i,l)$ and $u(v-i,l)$ are defined similarly to how we
defined them for $P(\filled)$ and $u(\filled)$. They correspond to the
probability and utilities received when we visit $i$ new (previously unfilled)
information states with $l$ remaining matches (including current one). It holds
that $P(v,l) + P(v-1,l) + P(v-2,l) + \dots + P(v-v,l) = 1$ as the probability of
reaching a terminal history within a match is equal to one.

By using the induction hypothesis \eqref{eq:ind-hypo} on terms $a_{x,l-1}\
\forall x < v$ we get a lower bound $a_{x,l-1} \geq -x \Delta$ on all of $x$. By
comparing these bounds we can deduce that $a_{v-1,l-1}$ lower bounds all of
$a_{x,l-1}$ with
\begin{align}
    \label{eq:avl-1}
    a_{v-1,l-1} \geq -(v-1)\Delta.
\end{align}

Using this bound in~\eqref{eq:avl} we get
\begin{equation}
\begin{aligned}
a_{v,l} \geq &P(v,l)(u(v,l)-\epsv + a_{v,l-1}) \\ 
    &+ P(v-1,l)(u(v-1,l)-\epsv -(v-1)\Delta) \\
    &+ P(v-2,l)(u(v-2,l)-\epsv -(v-1)\Delta) \\
    &+\ \dots \\
    &+ P(v-v,l)(u(v-v,l)-\epsv -(v-1)\Delta).
\end{aligned}
\end{equation}

We can factor it out as 
\begin{equation}
\begin{aligned}
    a_{v,l} \geq &P(v,l)(u(v,l)-\epsv + a_{v,l-1}) \\ 
    &+ (1-P(v,l))(-(v-1)\Delta-\epsv) \\
    &+ P(v-1,l)u(v-1,l) + P(v-2,l)u(v-2,l) \\
    &+ \dots\ + P(v-v,l)u(v-v,l).
\end{aligned}
\end{equation}

We replace the utilities $u(v-1), u(v-2),\ \dots,\ u(v-v)$ by $u(\notnash)$ from~\eqref{eq:ux-def}:
\begin{align}
    a_{v,l} \geq & P(v,l)(u(v,l)-\epsv + a_{v,l-1}) \\
    & + (1-P(v,l))(u(\notnash) - \epsv -(v-1)\Delta).
\end{align}

By using the induction hypothesis \eqref{eq:ind-hypo} we get
\begin{align}
    \label{eq:avl-ineq}
    a_{v,l} \geq & P(v,l)(u(v,l)-\epsv - v \Delta) \\
    & + (1-P(v,l))(u(\notnash) - \epsv -(v-1)\Delta).
\end{align}

Expanding the terms
\begin{equation}
\begin{aligned}
    a_{v,l} \geq & P(v,l)(-v \Delta) \\ 
    &+ (1-P(v,l))( -(v-1)\Delta) \\
    &+ P(v,l)(u(v,l)-\epsv) \\
    &+ (1-P(v,l))(u(\notnash) - \epsv)
\end{aligned}
\end{equation}

and using Lemma~\ref{lem:bound-eq-diff} with $\filled = v,l$ we have
\begin{equation}
\begin{aligned}
    a_{v,l} \geq & P(v,l)(-v \Delta) + (1-P(v,l))( -(v-1)\Delta) \\
    &- (1-P(v,l))(u(\notnash) - \epsv + \Delta) \\
    &+ (1-P(v,l))(u(\notnash) - \epsv).
\end{aligned}
\end{equation}

Simplifying, we get a bound on $a_{v,l}$:
\begin{align}
    a_{v,l} \geq - v \Delta.
\end{align}

Note that this bound holds also if $P(v,l)=1$ or $P(v,l)=0$:
\begin{itemize}
    \item $P(v,l)=0$: Then \eqref{eq:avl-ineq} becomes
        \[
            a_{v,l} \geq u(\notnash) - \epsv - (v-1)\Delta.
        \]
        Using the same argument as in~\eqref{eq:util-eq-ne},
        \[
            a_{v,l} \geq -\Delta - (v-1)\Delta = -v \Delta.
        \]
    \item $P(v,l)=1$: We use \eqref{eq:partition-eq-diff}, which becomes $u(v,l) - \epsv \geq 0$. Then \eqref{eq:avl-ineq} becomes
        \[
        a_{v,l} \geq u(v,l) - \epsv - v\Delta \geq - v\Delta.
        \]
\end{itemize}

Since at the beginning of the game-play there are $|\mc S_1|$ unfilled information states, we arrive at the original theorem
\begin{align*}
a_{|\mc S_1|,k} &\geq -|\mc S_1| \Delta.
\end{align*}
\end{proof}

\section{Experiment details} \label{ap:experiment}

As OOS runs MCCFR samples biased to particular information states, individual
MCCFR runs can converge to different $\epsilon$-equilibria, as the MCCFR is
parametrized differently in each information state. Additionally OOS runs in an
online setting, where the algorithm is given a time budget for computing the
strategy, and it may make different numbers of samples in each targeted
information state.

We emulate this experimentally by slightly modifying initial regrets to produce
distinct convergence trajectories. We show it is possible to highly exploit the
online algorithm: in fact, it is possible to exploit the algorithm more than the
worst of any individual biased strategies it produces, not just the expected
strategies. This modification is sound: the initial regrets will ``vanish'' over
longer sampling and the strategies will converge to an equilibrium in the limit.
This is justified by the MCCFR regret bound~\cite[Theorem~5]{MCCFR}.

We use two games: Coordinated Matching Pennies (CMP) from
Section~\ref{sec:online-algorithm} and Kuhn Poker~\cite{kuhn1950simplified}. We
use the no-memory online setting. Nash equilibria in both games are parametrized
with a single parameter $\alpha \in \langle 0, 1 \rangle$ for one player, while
the opponent has only a single unique equilibrium\footnote{ In CMP, $p=\alpha$
(playing Heads in $s_1$) and $q=1-\alpha$ (playing Heads in $s_2$). In Kuhn
Poker, constructing equilibrium strategy based on $\alpha$ is more complicated
and we refer the reader to~\cite{kuhn1950simplified}
or~\cite{hoehn2005effective} for more details.}. In both games, equilibria
require the strategies to be appropriately balanced, an effect of non-locality
problem~\cite{frank1998search} present only in imperfect information games. When
we compose the final strategy from partial online strategies, this balance can
be lost, resulting in high exploitability of the composed strategy.

We modify the initial regrets with following procedure:

\begin{itemize}

\item Choose a distinct value of $\alpha$, one for each of the player's top-most
information states in the game. Compute an equilibrium strategy according to
$\alpha$.

\item Directly copy the behavioral strategy into regret accumulators, and
multiply them by a constant $\mu$.

\end{itemize}

\begin{figure*}[t]
  \centering
  \includegraphics[width=0.9\textwidth]{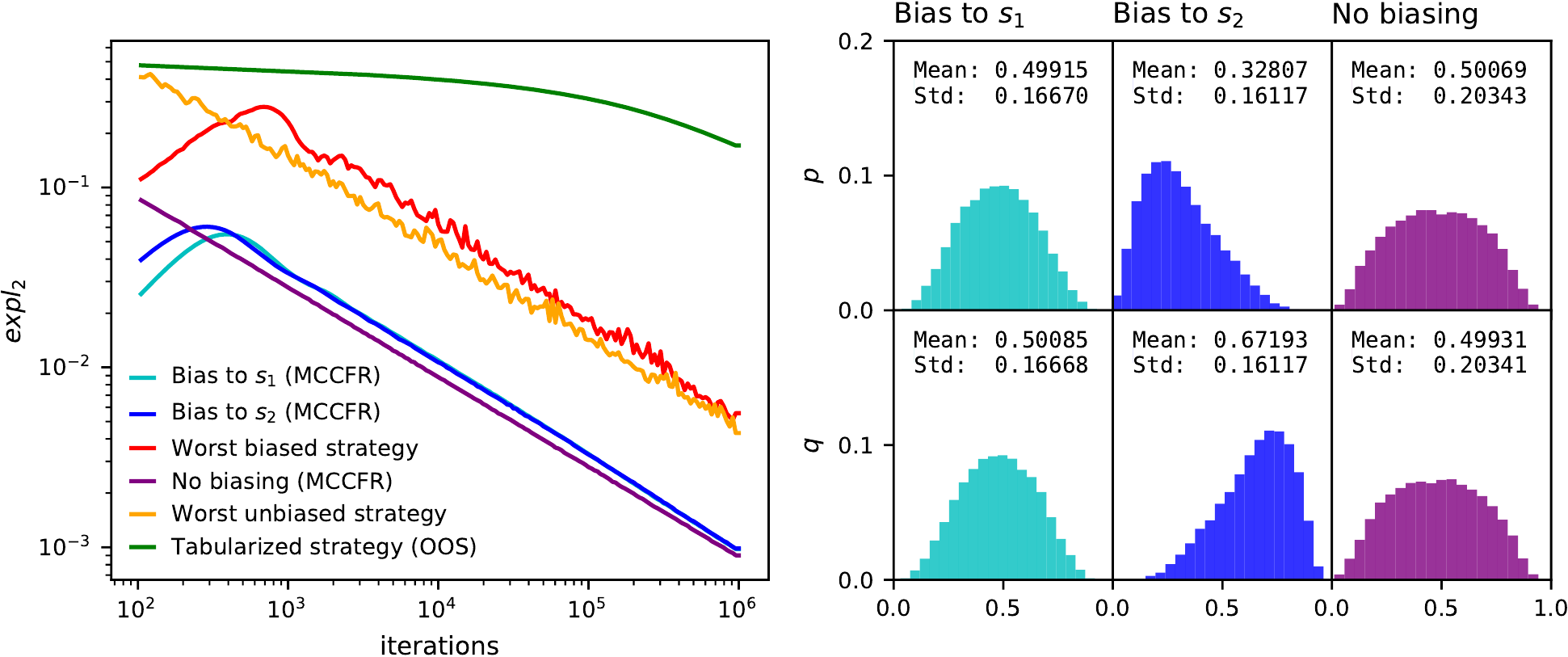}

  \caption{ Coordinated Matching Pennies. Left: While individual MCCFR
  strategies have low exploitability of $\sim10^{-3}$, the~tabularized OOS
  strategy has high exploitability of $0.17$ after $10^6$~iterations. Right:
  Normalized histograms of probabilities of playing Heads in $s_1$ - $p$ and
  $s_2$ - $q$ after $10^6$~iterations. The histograms within columns are
  correlated, as they approximately satisfy equilibrium condition $p+q=1$.
  Tabularized strategy, combination of $(p, s_1)$ and $(q, s_2)$ violates this
  constraint, resulting in high exploitability. }
  \label{fig:oos-mp}
\end{figure*}

This simple procedure effectively kick-starts the algorithm to produce distinct
trajectories based on $\alpha$.

\begin{figure}
  \centering
  \includegraphics[width=0.4\textwidth]{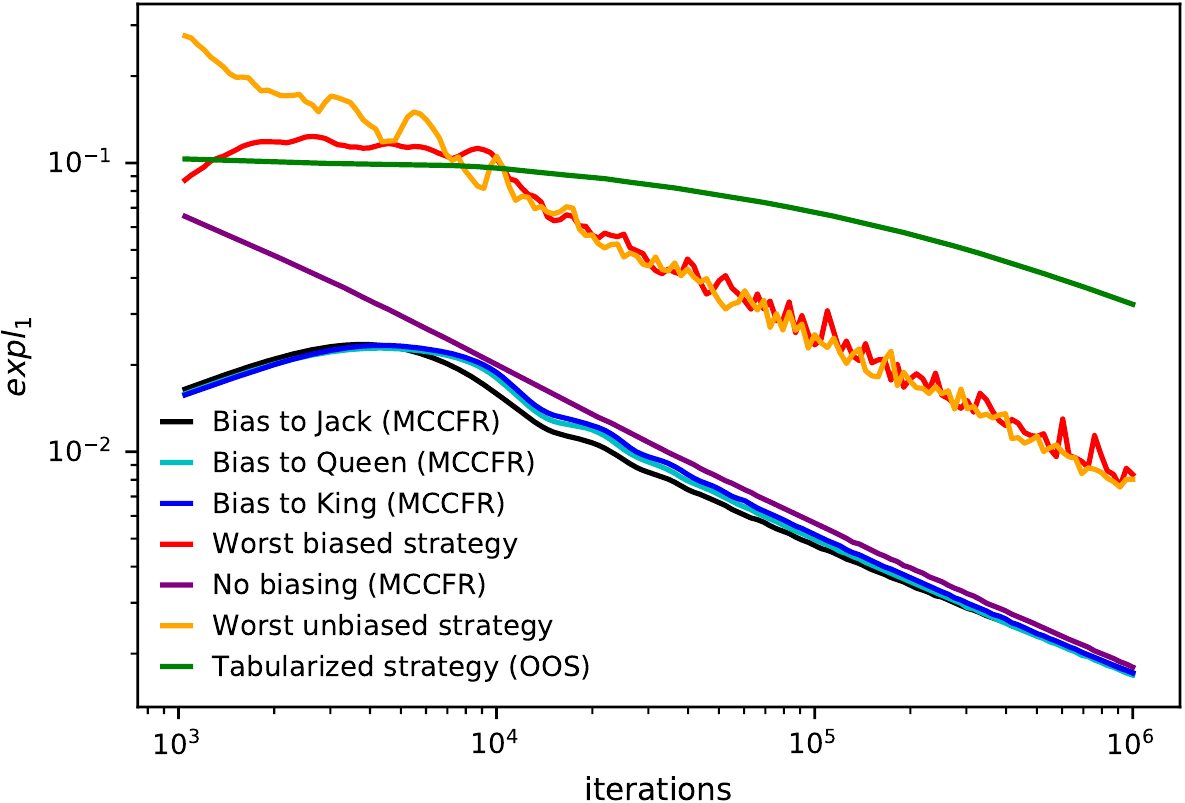}
  \caption{Kuhn Poker}
  \label{fig:oos-kuhn}
\end{figure}

In~\fig{oos-mp} and in~\fig{oos-kuhn}, we show that individual biased strategies
converge to Nash equilibria, but the tabularized strategy has higher
exploitability even than the worst individual strategy. In CMP, we bias the
second player to play in information state $s_1$ ($\alpha=0.5$) or $s_2$
($\alpha=1$) information states. In Poker, we bias the first player to play Jack
($\alpha=0$), Queen ($\alpha=1/2$) or King ($\alpha=1$) card. For both
experiments, exploration was set to 0.6, biasing to 0.1, and $\mu=500$, a small
regret that can be accumulated after less than 500 samples. Within our online
framework, the state $\searchparam$ consists of regrets and average strategy
accumulators for all information states, and from the state of the pseudo-random
number generator, which has distinct initial seeds for each match. The expected
strategies are estimated as an average over $3\cdot10^4$ seeds. We plot the
worst strategy from these individual biased strategies over all the seeds for
all iterations. We plot also MCCFR strategy for reference, to see the influence
of biasing and regret initialization.

\end{document}